%% file: main.tex
\let\origvec\vec
\documentclass[10pt,conference,compsocconf,letterpaper]{IEEEtran}
\let\spvec\vec

\usepackage[T1]{fontenc}

\let\vec\origvec
\usepackage{amsmath}
\let\vec\spvec

\usepackage{amsfonts}
\usepackage{amssymb}
\usepackage{pstricks}
\usepackage{pst-tree}
\usepackage{subfigure}
\usepackage{tikz}
\usetikzlibrary{shapes,snakes,positioning}

\usepackage{mathtools}
\usepackage{booktabs}
\usepackage[font=small,format=plain,labelfont=bf,up,textfont=up]{caption}

\usepackage{subfigure}
\usepackage{graphicx}
\graphicspath{{figures/}}

\newtheorem{theorem}{Theorem}
\newtheorem{lemma}[theorem]{Lemma}
\newtheorem{proposition}[theorem]{Proposition}

\newtheorem{definition}{Definition}

\newtheorem{remark}{Remark}

\newenvironment{proof}[1][Proof]{\begin{trivlist}
\item[\hskip \labelsep {\bfseries #1}]}{\end{trivlist}}
\newcommand{\qed}{\nobreak \ifvmode \relax \else
      \ifdim\lastskip<1.5em \hskip-\lastskip
      \hskip1.5em plus0em minus0.5em \fi \nobreak
      \vrule height0.75em width0.5em depth0.25em\fi}

\newcommand{\conf}[1]{\ensuremath\mathcal{C}(#1)}
\newcommand{\confext}[1]{\ensuremath\tilde{\mathcal{C}}(#1)}
\newcommand{\children}[1]{\ensuremath\mathcal{S}_{#1}}
\newcommand{\subtree}[1]{\ensuremath\mathcal{S}^\star_{#1}}
\newcommand{\maxheightchildren}[1]{\ensuremath\mathcal{\hat S}_{#1}}

\newcommand{\calO}{\ensuremath\mathcal{}}

\IEEEoverridecommandlockouts

\begin{document}

\title{Self-Stabilizing Balancing Algorithm\\ for Containment-Based Trees%
\thanks{This research was initiated while Evangelos Bampas was with 
      MIS Lab. University of Picardie Jules Verne, France. It was partially funded by 
      french National Research Agency (08-ANR-SEGI-025). 
      Details of the project on http://graal.ens-lyon.fr/SPADES.
}
}
% \title{Self-Stabilizing Balancing Algorithm for Containment-Based Trees\thanks{Grants here.}}

% ICDCS

% \author{
% \IEEEauthorblockN{Evangelos Bampas}
% \IEEEauthorblockA{School of Elec.\ \& Comp.\ Eng.,\\
% National Technical Univ.\ of Athens,\\
% 15780 Zografou, Greece\\
% Email: ebamp@cs.ntua.gr}
% \and
% \IEEEauthorblockN{Anissa Lamani}
% \IEEEauthorblockA{}
% \and
% \IEEEauthorblockN{Franck Petit}
% \IEEEauthorblockA{}
% \and
% \IEEEauthorblockN{Mathieu Valero}
% \IEEEauthorblockA{Orange Labs\\
% 92130 Issy-les-Moulineaux, France\\
% Email: mathieu.valero@lip6.fr}}

%\author{Evangelos Bampas, Anissa Lamani\thanks{LIP6, 4 place Jussieu, F-75005 Paris, France. Fax: +33 144 277 495.
%e-mail: first.last@lip6.fr}, Franck Petit, Mathieu Valero
%\thanks{Orange Labs, 38 rue du g\'en\'eral Leclerc, 92130 Issy-les-Moulineaux, France. e-mail: first.last@lip6.fr}
%}
%\institute{
%MIS Lab. University of Picardie Jules Verne, France \and%
%LIP6 CNRS UMR 7606 - INRIA - UPMC Sorbonne Universities, France%
%}

\author{\IEEEauthorblockN{Evangelos Bampas\IEEEauthorrefmark{1}, Anissa Lamani\IEEEauthorrefmark{2},
Franck Petit\IEEEauthorrefmark{3}, and Mathieu Valero\IEEEauthorrefmark{4}}%
\IEEEauthorblockA{\IEEEauthorrefmark{1}School of Electrical and Computer Engineering,
National Technical University of Athens\\
15780 Zografou, Greece\\ 
Email: ebamp@cs.ntua.gr}%
\IEEEauthorblockA{\IEEEauthorrefmark{2}MIS Lab. University of Picardie Jules Verne, France\\
Email: anissa.lamani@u-picardie.fr}%
\IEEEauthorblockA{\IEEEauthorrefmark{3}LIP6 CNRS UMR 7606 - INRIA - UPMC Sorbonne Universities, France\\
Email: Franck.Petit@lip6.fr}% 
\IEEEauthorblockA{\IEEEauthorrefmark{4}Orange Labs\\
92130 Issy-les-Moulineaux, France\\
Email: mathieu.valero@lip6.fr}}

% ICDCS

\maketitle

\input{abstract}

\input{introduction}

% MOVED IN THE INTRODUCTION \input{related-work}

\input{primitive}

\input{model}

\input{Solution.tex}

%Now a subsection

\input{contribution}

\input{simulation}

\input{conclusion}

%\bibliographystyle{abbrv}

\bibliographystyle{IEEEtran}
\bibliography{main}

%\newpage

\appendix

\input{appendix}

\end{document}

%% file: abstract.tex
\begin{abstract}
Containment-based trees encompass various handy structures such as B+-trees, R-trees and M-trees.
They are widely used to build data indexes, range-queryable overlays, publish/subscribe systems both in centralized and
distributed contexts.
In addition to their versatility, their balanced shape ensures an overall satisfactory performance.
Recently, it has been shown that their distributed implementations can be %surprisingly 
fault-resilient.
However, this robustness is achieved at the cost of unbalancing the structure.
While the structure remains correct in terms of searchability, its performance can be significantly decreased.
In this paper, we propose a distributed self-stabilizing algorithm to balance containment-based trees.
\end{abstract}

\begin{IEEEkeywords}
self-stabilization, balancing algorithms, containment-based trees
\end{IEEEkeywords}

%% file: introduction.tex
\section{Introduction}
\label{sec:introduction}

\iffalse
In the former all invariants are ensured by the $join$ and $leave$ primitives.
Basically the $join$ consists in inserting leaves in the distributed tree while the $leave$ consists in removing leaves and/or collapsing nodes.
In case of faults, the tree is basically rebuild.
If ``too many'' children of a node crash, it no longer verifies {\it bounded node degree} and thus collapse.
Moreover each node maintains a list of ---some of--- its ancestors further than its father.
When a node detects a fault of one of its ancestor, it propagates a reconnection message in its subtree indicating an ancestor closer to the root than the faulty one.
All leaves of the subtree receiving a reconnection message will re $join$ the overlay accordingly.
\fi

\iffalse
The demand to handle flexibility and efficiency of data discovery in distributed systems led to the development of various overlay structures. These structures are used for indexing purposes in databases and P2P systems. Several types of overlay structures, such as tree overlays, have been proposed in the literature to obtain both semantic ability and scalability. Among the proposed tree overlays there are overlays that are based on containment relation. These structures allow to support range search in addition of key-based exact match search.
\fi

Several tree families are based on a containment relation. Examples include B+-trees~\cite{Btree}, R-trees~\cite{Rtree},
and M-trees~\cite{Mtree}. 
They are respectively designed to handle intervals, rectangles, and balls.
%and are usually used for indexing purposes in databases and P2P systems.
%They allow to efficiently and easily support automatic-completion and range queries by knowing exactly where to search.%support exact match and range queries.
Their logarithmic height ensures good performance for basic insertion/deletion/search primitives.
Basically they rely on a partial order on node labels.
They can be specified as follows:
\begin{enumerate}
	\item {\em Tree nature}. The graph is acyclic and connected.
	\item {\em Containment relation}. Every non-root node $n$ satisfies $label(n) \sqsubseteq label(father(n))$.
	\item {\em Bounded degrees}. The root has between 2 and $M$ children, each internal node has between $m$ and $M$ children ($M \geq 2m$).
	\item {\em Balanced shape}. All leaves are at the same level.
\end{enumerate}
In a distributed context, no node has a global knowledge of the system as each node has only access to its local information.
As a consequence, the aforementioned invariants should be expressed as ``local'' constraints; at the level of a node.
Operations preserving or restoring those invariants should also be ``as local as possible''.

Preserving the {\em tree nature} has been addressed in previous work~\cite{DRtree,VBI,PexpanderVtree}.
It is tightly related to the distribution model of the structure, the centralized case being far less stressing.
The {\em containment relation} is easy to preserve as node labels can be ``enlarged'' or ``shrunk''.
The {\em bounded node degrees} are ensured with $split$ or $collapse$ primitives when a node has too many
or too few children, respectively.
%Those primitives are inate to all aforementioned tree families (and also implemented in several other trees such as Q-trees~\cite{QuadTree} and KD-trees~\cite{KDtree}).
The {\em balanced shape} of the tree is especially important in terms of performance; in conjunction with node
degree bounds, it ensures that the tree has logarithmic height. A number of approaches have been proposed in the
literature in order to balance tree overlays. However, these solutions have many limitations in particular when applied
to containment-based tree overlays.  

The works in~\cite{DRtree,PexpanderVtree} have the extra property of being \emph{self-stabilizing}.
A self-stabilizing system, as introduced in~\cite{D74}, is guaranteed to converge to the intended behavior in finite
time, regardless of the initial states of nodes.
Self-stabilization~\cite{D74} is a general technique to design distributed systems that can handle arbitrary transient
faults. 

% ICDCS

Figure~\ref{fig:model} shows the overlay lifecycle borrowed from~\cite{PexpanderVtree}.
Rectangles refer to states of the distributed tree.
Transitions are labelled with events ({\it join} and {\it faults}) or algorithms ({\it repair} and {\it balance}) triggering them~\footnote{A balanced tree might remain balanced or might become searchable in case of {\it faults}, while it might remain balanced in case of {\it joins}.
However for the sake of readability only most stressing ---which are also the most likely--- transitions are shown.}.
Initially the tree is empty, then some peers join the overlay building a searchable tree.
The balancing algorithm eventually balances the tree.
In case of {\it joins} the tree may become unbalanced.
In case of {\it faults} the tree may become disconnected.
The {\it repair} algorithm eventually reconnects the tree and fixes containment relation.
%The overlay has three states; balanced when all invariants are ensured, searchable when all but the fourth invariant are ensured, and disconnected when at least {\it tree nature} invariant is violated.
The distinction between searchable and balanced states emphasizes a separation of concerns between correction and
performance.
Basically {\it repair} is about correction, while {\it balance} is about performance.
It is also interesting to point out that if the {\it balance} algorithm is self-stabilizing, then the {\it join}
algorithm does not have to deal with performance.
Moreover, if the {\it repair} algorithm is self-stabilizing, a peer quitting the overlay does not have to communicate
with any other peer; departures can be handled as faults.
\begin{figure}[bt]

\centering

\tikzstyle{action}=[]
\tikzstyle{state}=[rectangle,draw,thick]

\iffalse
\subfigure[Borrowed from~\cite{DRtree}]{
\begin{tikzpicture}[scale=0.2]
	\node[action]			(join)		 						{join/leave};
	\node[state]			(balanced)		[below of=join] 	{Balanced};
	\node[state]			(disconnected)	[below of=balanced,node distance=1.5cm]	{Disconnected};
	\draw[->]				(join.south) -- (balanced.north);
	\draw[->,snake=snake]	(balanced.south) -- (disconnected.north) node [left,text centered,midway]{{\it faults}};
	\draw[-]				(disconnected.east) -- ++(.5,0) |- (join.east);
\end{tikzpicture}
\label{fig:old_model}
}
\fi

\begin{tikzpicture}
	\node[state]			(balanced)									
		{Balanced};
	\node[state]			(disconnected)	[below=2.5cm of balanced]			
{Disconnected};
	\node[state]			(searchable)	[below left=1cm and 0cm of balanced]	{Searchable};
	\node[action]			(join)		 	[left=0.5cm of searchable]			
{join};
	\node[action]			(balance)		[above of=searchable]	{balance};
	\node[action]			(repair)		[below of=searchable]	{repair};	
	\draw[->]				(join) -- (searchable);
	\draw[->,snake=snake]	(balanced) -- (disconnected) node [right,text centered,midway]{{\it faults}};
	\draw[-]				(disconnected) -| (repair);
	\draw[->]				(repair) -- (searchable);
	\draw[-]				(searchable) -- (balance);
	\draw[->]				(balance) |- (balanced);
	\draw[-,snake=snake]	(searchable.east) -- ++(0.8,0);
\end{tikzpicture}

%\label{fig:new_model}
\caption{Containment-based tree overlay lifecycle.
The overlay is balanced when all invariants are ensured, searchable when all but the fourth invariant are ensured, and
disconnected when at least the {\it tree nature} invariant is violated.}
\label{fig:model}
\end{figure}
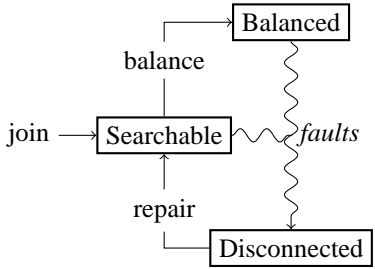

\input{related-work}

\subsection{Contribution}

In this paper, we first show that edge swapping can be used as a balancing primitive.
Then we propose a distributed self-stabilizing algorithm balancing any containment-based tree,
such as B+-trees~\cite{Btree}, R-trees~\cite{Rtree} and M-trees~\cite{Mtree}.
Our algorithm can be used in~\cite{SDRtree,PexpanderVtree} to enhance the performance of repaired trees, or
in~\cite{VBI,BATON} as a ``core'' balancing mechanism.
We further prove the correctness of the algorithm and 
investigate its practical convergence speed via simulations.

\subsection{Roadmap}

%In Section~\ref{sec:related-work}, we discuss several distributed implementations of containment-based trees. 
In Section~\ref{sec:primitive}, we argue that edge swapping is practically better suited than rotations to balance
containment-based trees. In Section~\ref{sec:model}, we present the model we use to describe our algorithm. In
Section~\ref{sec:algo}, we propose a distributed self-stabilizing algorithm that relies on edge swapping to balance
any containment-based tree. %In Section~\ref{sec:proof}, we prove the termination and
In the same section, we prove the termination and 
correctness of the algorithm. In Section~\ref{sec:simulation}, we investigate the practical termination time of our
algorithm via simulations. Section~\ref{sec:conclusion} contains some concluding remarks and possible directions
for future work.

\iffalse
\paragraph{Distributed R-tree.}

A distributed R-tree is a R-tree which nodes are stored in the memory of several computers.
In such context, in case of faults (computer crashes), the tree is very likely to become disconnected.
Different distributed algorithms have been proposed to repair the structure~\ref{Sylvia,DISC2010}.
The reparation has been specified in~\ref{DISC2010} as a process starting from a forest of connected components eventually leading to a tree.
More precisely, it emphasizes three steps:
\begin{enumerate}
	\item restoring connectivity of the structure
	\item restoring containment relation
	\item balancing the structure
\end{enumerate}

Step one and two deals with correction of the structure.
At the end of the second step, the structure is said ``searchable''.
Almost every invariants of R-trees are verified.
Constraints on nodes degrees are ensured using $split$ and $collapse$ operations.
Containment relation is ensured.
The only problem is that all leaves are not at the same level; the tree is thus not balanced wvvhp~\footnote{With very very high probability}.

Step three deals with performances.
This article investigates a distributed self-stabilizing balancing mechanism for R-trees (that could also be used to balance containment based structures notably B-trees, M-trees).
\fi

%% file: related-work.tex
\subsection{Related Work}

DR-tree~\cite{DRtree} is a distributed version of R-trees~\cite{Rtree} developed to build a brokerless peer to peer
publish/subscribe system.
Each computer stores a leaf and some of its consecutive ancestors.
The closer node to the root stored by a computer is called its ``first node''.
Each computer knows some ancestors of its first node; from its grand father to the root.
The tree may become unbalanced in case of faults.
When a computer detects that the computer storing the father of its first node is crashed, it broadcasts a message in
the subtree starting from its first node.
All leaves belonging to that subtree will be reinserted elsewhere starting from an ancestor of the crashed peer's first
node.
As a consequence, no balancing primitive is used.
However, up to half of the nodes of the tree may have to be reinserted (if faults occur near the root).
%It has been shown in~\cite{SSS2010} that 

SDR-tree~\cite{SDRtree} is an indexing structure distributed amongst a cluster of servers.
Its main aim is to provide efficient and scalable range queries.
It builds balanced full binary trees satisfying containment relation.
Each computer stores exactly one leaf and one internal node.
Instead of classic $split$ and $collapse$ operations, this paper extensively uses subtree heights to guarantee the
balanced shape of the structure.
They adapt the concept of rotation~\cite{AVL} to multidimensional data.
However they don't mention how to repair the structure in case of crashes.

VBI~\cite{VBI}, a sequel of BATON~\cite{BATON}, is a framework to build distributed spatial indexes on top of binary trees.
It provides default implementations for all purely structural concerns of binary trees: maintaining father and children
links, rotations, etc.
Developers only have to focus on several operations, mostly related to node labeling (such as $split$ and $collapse$).
However, despite their interesting approach of mapping any tree on a binary one, they do not provide a handy way to
parameterize the system.
While modifying node degree bounds should be a simple way to tune system performances, the fixed degree of the core
structure (a tradeoff between reusability and performance) cannot be bypassed.

The solution proposed in~\cite{PexpanderVtree} also deals with full binary trees satisfying containment relation.
The distribution is the same as in the SDR-tree; each computer stores exactly one leaf and one internal node.
The core contribution of this paper is to prove that if the leaf and the internal node held by each machine are randomly
chosen, then the graph between computers (namely, the communication graph) is very unlikely to be disconnected in
case of faults.
More precisely, faults disconnect the logical structure, but not the communication graph.
This paper proposes an algorithm exploiting the connectivity of the communication graph to repair the
logical structure restoring the {\it tree nature}, {\it containment relation} and {\it bounded node degrees}.
This approach is shown as much cheaper than~\cite{DRtree}, both in terms of recovery time and cost in messages.
However the restored tree can be slightly unbalanced and the restoration of its shape is not addressed in the paper.

The algorithm proposed in~\cite{PexpanderVtree} has the desirable property of being self-stabilizing. 
There exist several distributed, self-stabilizing algorithms in the literature that maintain a global property based 
on node labels, {\em e.g.}, \cite{HM01a,BDV05,DLPT}.  The solutions in \cite{HM01a,BDV05} achieve 
the required property (resp. Heap and BST) by reorganizing the key in the tree.  They have 
no impact on the tree topology.  The solution in \cite{DLPT} arranges both node labels (or, keys) 
and the tree topology.  None of the above self-stabilizing solutions deal with the balanced property.

%Tedeschi/Franck etc\dots\ distributed prefix tree/B+-tree~\cite{Btree} ?

%MOVED IN THE CONTRIBUTION PARAGRAPH
%-In this paper, we propose a distributed self-stabilizing algorithm to rebalance trees based on a containment relation,
%-such as B+-trees~\cite{Btree}, R-trees~\cite{Rtree} and M-trees~\cite{Mtree}.
%-%In~\cite{SDRtree,PexpanderVtree} each non leaf node has exactly two children.
%-
%-
%-%Wen are more general as we make no assumption on the maximum degree of nodes; each non leaf node has at least two children.
%-% Our distribution assumption are thus slightly different; each computer store exactly one leaf and up to one internal node.
%-Our algorithm could be used in~\cite{SDRtree,PexpanderVtree} to enhance the performance of repaired trees, or
%-in~\cite{VBI,BATON} as a ``core'' balancing mechanism.

%% file: primitive.tex
\section{Balancing Primitive}
\label{sec:primitive}

%The shape of a tree is a purely structural concept.
Let $\alpha \in \mathbb{N}$.
A tree is balanced iff all its nodes are balanced.
A node is balanced iff the heights of any pair of its children differ at most by $\alpha$.
In the remainder of this paper, we make two assumptions:
\begin{itemize}
	\item $\alpha \geq 1$.
	\item each non-leaf node has at least two children.
\end{itemize}
The assumption on~$\alpha$ is weaker than those of~\cite{SDRtree,PexpanderVtree,VBI} and still ensures logarithmic
height of the tree~\cite{DBLP:journals/tc/LuccioP76}.
The degree assumption is also weaker than those of~\cite{SDRtree,PexpanderVtree,VBI} and allows more practical tree
configurations.

Basically, a balancing primitive is an operation that eventually reduces the difference between the heights of some
subtrees by modifying several links of the tree.
Those link modifications may have some semantic impact if they break node labeling invariants.

%{\it and dodges several pathologic also dodged by quoted works.}

\subsection{Rotation} In BST~\cite{BST} and AVL~\cite{AVL}, the well known rotation primitive is used to ensure the
balanced shape of the tree.
%Their structural impact consists in ``exchanging'' subtrees by modifying several links.
%They have no semantic impact on BST; indeed they preserve the total ordering relation of nodes labels.
However, when dealing with structures relying on partially ordered data, rotations do have a semantic impact.
%While their impact on height equilibrium is still interesting, they have two drawback.
%First, they may violate containment relation.
%Second, they modify ``more links than required'' in the sense that other primitives breaking containment relation have the same (or better) impact on height with less links modifications.
Moreover, in a distributed context, if a node~$n$ and its father or grandfather concurrently execute rotations, they may
both ``write'' $father(n)$.
It follows that the use of distributed rotations requires synchronization to preserve {\it tree structure}.

\iffalse
\begin{figure}[tpb]
\centering
\subfigure[Before rotation of $a$]{
\begin{pspicture}(0,-3)
\pstree[nodesep=1pt,treesep=.7,levelsep=.7]{\Tr{ }}{
	\pstree{\Tcircle{a}}{
		\pstree{\Tcircle{b}}{
			\pstree{\Tcircle{e}}{\Tfan}
			\pstree{\Tcircle{f}}{\Tfan}
		}
		\pstree{\Tcircle{c}}{\Tfan}
		\pstree{\Tcircle{d}}{\Tfan}
	}
}
\end{pspicture}

}
\subfigure[After rotation of $a$]{
\begin{pspicture}(0,-3)
\pstree[nodesep=1pt,treesep=.7,levelsep=.7]{\Tr{ }}{
	\pstree{\Tcircle{b}}{
		\pstree{\Tcircle{f}}{\Tfan}
		\pstree{\Tcircle{a}}{
			\pstree{\Tcircle{c}}{\Tfan}
			\pstree{\Tcircle{d}}{\Tfan}			
			\pstree{\Tcircle{e}}{\Tfan}
		}
	}
}
\end{pspicture}
}
\caption{\label{fig:rotation}Rotation exemple}
\end{figure}
\fi

%Figure~\ref{fig:rotation} shows the impact of a left rotation on a node of a non binary tree.
%In structural terms, the fathers of $a$, $b$ and $e$ are modified.
%In semantic terms, assuming a containment based tree, $label(b)$ needs to be updated.

\subsection{Edge swapping}
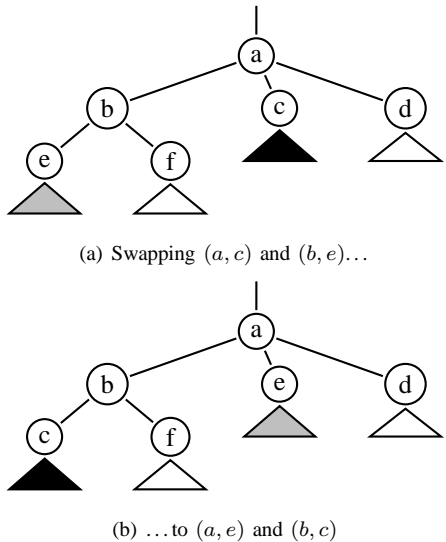
\begin{figure}[tpb]
\centering
\subfigure[Swapping $(a,c)$ and $(b,e)$\dots]{
\begin{pspicture}(0,-3)
\pstree[nodesep=1pt,treesep=.7,levelsep=.7]{\Tr{ }}{
	\pstree{\Tcircle{a}}{
		\pstree{\Tcircle{b}}{
			\pstree{\Tcircle{e}}{\Tfan[fillstyle=solid,fillcolor=lightgray]}
			\pstree{\Tcircle{f}}{\Tfan}
		}
		\pstree{\Tcircle{c}}{\Tfan[fillstyle=solid,fillcolor=black]}
		\pstree{\Tcircle{d}}{\Tfan}
	}
}
\end{pspicture}
\iffalse
\begin{pspicture}(0,-3)
\pstree[nodesep=1pt,treesep=.7,levelsep=.7]{\Tcircle{}}{
	\pstree{\Tcircle{a}}{
		\pstree{\Tcircle{b}}{
			\pstree{\Tcircle{e}}{
				\Tcircle{}
				\Tcircle{}
			}
			\pstree{\Tcircle{f}}{\Tfan}
		}
		\pstree{\Tcircle{c}}{
			\Tcircle{}
			\Tcircle{}		
		}
		\pstree{\Tcircle{d}}{\Tfan}
	}
	\Tcircle{}
}
\end{pspicture}
\fi
}
\subfigure[\dots to $(a,e)$ and $(b,c)$]{
\begin{pspicture}(0,-3)
\pstree[nodesep=1pt,treesep=.7,levelsep=.7]{\Tr{ }}{
	\pstree{\Tcircle{a}}{
		\pstree{\Tcircle{b}}{
			\pstree{\Tcircle{c}}{\Tfan[fillstyle=solid,fillcolor=black]}			
			\pstree{\Tcircle{f}}{\Tfan}
		}
		\pstree{\Tcircle{e}}{\Tfan[fillstyle=solid,fillcolor=lightgray]}
		\pstree{\Tcircle{d}}{\Tfan}
	}
}
\end{pspicture}
}
\caption{\label{fig:edge-swap}An example of edge swapping. The fathers of~$c$ and~$e$ are modified. The label of~$b$
needs to be updated.}
\end{figure}
%Swapping two edges is a common operation in directed graph.
Given two edges~$(a,b)$ and~$(c,d)$, swapping them consists in exchanging their tails (resp.\ heads).
Formally, $swap((a,b),(c,d))$ modifies the edges of the graph as follows: $E(G) \coloneqq
E(G)-\{(a,b),(c,d)\}+\{(a,d),(c,b)\}$. Figure~\ref{fig:edge-swap} contains an illustration.
With the algorithm that we present in Section~\ref{sec:algo}, concurrent swaps cannot conflict.
The use of this balancing primitive is thus more suitable in a distributed context.

%Figure~\ref{fig:edge-swap} shows the impact of edge swapping on a tree.
%$a$ is the source of the swap while $b$ is the target.
%In structural terms, the fathers of $c$ and $e$ are modified.
%In semantic terms, assuming a containment based tree, $label(b)$ needs to be updated.

%\paragraph{Conclusion}
%Edge swapping has the same impact on heights and labels than rotations.
%However, they require less links modifications.
%{\it Moreover, we'll show in section~\ref{sec:contribution} that edge swapping requires less synchronisation than rotations.
%This is a critical point in distributed context.}

%% file: model.tex
\section{Model}
\label{sec:model}
%%%%%%%%%%%%%%%%%%%%%%%%%%%%%%%%%%%%%%%%%%%%%%%%%%%%%%%%%%%%%%%%%%

In this paper, we consider the classical local shared memory model, known as the state model, that was introduced by
Dijkstra~\cite{D74}. In this model, communications between neighbours are modeled by direct reading of variables
instead
of exchange of messages. The program of every node consists in a set of shared variables (henceforth referred to as
variable) and a finite number of actions. Each node can write in its own variables and read its own variables and those
of its neighbors. Each action is constituted as follows:

\begin{center} $<Label>::<Guard>$  $<Statement>$ \end{center}

The guard of an action is a boolean expression involving the variables of a node~$u$ and its neighbours. The statement
is an action which updates one or more variables of~$u$. Note that an action can be executed only if its guard is true.
Each execution is decomposed into steps. 

The state of a node~$u$ is defined by the value of its variables. It consists of the following pieces of information:
\begin{itemize}
  \item An integer value which we call the \emph{height value} or \emph{height information} of node~$u$.
  \item Two arrays that contain the IDs of the children of~$u$ and their height values.
\end{itemize}
%The state of each node is stored in multiple-reader single-writer registers. We assume that each node has full
%access to its own state and read access to the state of its children. 
%
For the sake of generality, we do not make any
assumptions about the number of bits available for storing height information, thus the height value of a node can be an
arbitrarily large (positive or negative) integer.

The \emph{configuration} of the system at any given time~$t$ is the aggregate of the states of the individual
nodes. We will sometimes use the term ``configuration'' to refer to the rooted tree formed by the nodes.
\begin{definition}
 We denote by~$\conf{t}$ the configuration of the system at time~$t\geq 0$. For a node~$u$, we
denote by
$h_u(t)$ the value of its height variable in~$\conf{t}$, by $h^\star_u(t)$ its actual height in the tree in~$\conf{t}$,
by $\children{u}(t)$ the set of children of~$u$ in~$\conf{t}$, and by $\subtree{u}(t)$ the set of nodes in the subtree
rooted at~$u$ in~$\conf{t}$.
\end{definition}

%At any time~$t$, each node~$u$ is able to evaluate the following functions and predicates for itself and its children:
%\begin{itemize}
%  \item $\mathrm{max}(x)$: returns any $v\in\children{x}(t)$ such that $h_v(t) = \max_{w \in
%\children{x}(t)} h_w(t)$. If $x$ is a leaf, returns~$\bot$ (undefined).
%  \item $\mathrm{min}(x)$: returns any $v\in\children{x}(t)$ such that $h_v(t) = \min_{w \in
%\children{x}(t)} h_w(t)$. If $x$ is a leaf, returns~$\bot$.
%  \item $\mathrm{stable}(x)$: returns true if and only if $h_x(t) = 1 + \max_{w \in \children{x}(t)} h_w(t)$. If $x$ is
%a leaf, returns true if and only if $h_x(t) = 0$.
%  \item $\mathrm{balanced}(x)$: returns true if and only if for all $z,z'\in\children{x}(t)$,
%$|h_z(t)-h_{z'}(t)|\leq\alpha$.
%\end{itemize}
%
%When there are more than one possible return values for~$\mathrm{max}(x)$ and~$\mathrm{min}(x)$, an arbitrary choice is
%made. For simplicity, we can consider that the candidate node with the smallest ID is returned, although this will not
%be crucial for our results.

Let $\conf{t}$ be a configuration at instant $t$ and let $I$ be an action of a node~$u$. $I$ is {\em enabled} for~$u$
in~$\conf{t}$ if and only if the guard of~$I$ is satisfied by~$u$ in~$\conf{t}$. Node~$u$ is enabled in~$\conf{t}$ if
and only if at least one action is enabled for~$u$ in~$\conf{t}$. Each step consists of two sequential phases
 executed atomically:
($i$) Every node evaluates its guard;
($ii$) One or more enabled nodes execute their enabled actions. 
When the two phases are done, the next step begins. 
This execution model is known as the \emph{distributed daemon}~\cite{BGM89}. 
To capture asynchrony, we assume a semi-synchronous scheduler which picks any non-empty subset of the enabled nodes in
the current configuration and executes their actions simultaneously. We do not make any fairness assumptions, thus the
scheduler is free to effectively ignore any particular node or set of nodes as long as there exists at least one
other node that can be activated. 

%%%%%%%%%%%%%%%%%%%%%%%%%%%%%%%%%%%%%%%%%%%%%%%%%%%%%%%%%%%%%%%%%%
%\newcounter{action}
%\newcommand{\GA}[2]{\stepcounter{action}{\bf G\theaction}  #1 & {\bf A\theaction} #2}
%
%
%Each node~$u$ executes the following algorithm:
%
%\vspace{1ex}
%\noindent\setcounter{action}{0}
%\begin{tabular}{@{}ll@{}} \toprule
%{Guard} & {Action} \\ \midrule
%\GA{$\neg\mathrm{stable}(u)$}{$h_u \coloneqq 1 + \max_{w \in \children{u}} h_w$ (or $h_u \coloneqq 0$ if
%$u$ is a leaf)} \\[1ex]
%\GA{\begin{minipage}[t]{0.33\textwidth}
%$\mathrm{stable}(u) \wedge \mathrm{stable}(\mathrm{max}(u)) \wedge$ \\
%$\neg\mathrm{balanced}(u)$
%\end{minipage}}
%{swap edges $(u,\mathrm{min}(u))$ and $(\mathrm{max}(u),\mathrm{max}(\mathrm{max}(u))$} \\
%\bottomrule
%\end{tabular}
%
%\vspace{1ex}
%\noindent The value of $\mathrm{stable}(\bot)$ is false.
% We assume that the evaluation of the guards and the execution of the actions are performed
%atomically. %Note
%that the guards~\textbf{G1} and~\textbf{G2} are mutually exclusive and, by definition, the height update action
%effectively has priority over the edge swapping action. We say that a node is \emph{enabled for a height update} if
%\textbf{G1} is true, or \emph{enabled for a
%swap} if \textbf{G2} is true.
%
%When a node~$u$ performs a swap, four nodes are involved: $u$ itself, $\mathrm{max}(u)$, $\mathrm{min}(u)$, and
%$\mathrm{max}(\mathrm{max}(u))$. We refer to these nodes as the \emph{source}, the \emph{target}, the \emph{swap-out},
%and the \emph{swap-in} nodes of the swap, respectively.

\begin{definition}
 We refer to the activation of a non-empty subset~$A$ of the enabled nodes in a given configuration as an
\emph{execution step}. If $\mathcal{C}$ is the configuration of the system before the activation of~$A$ and
$\mathcal{C}'$ is the resulting configuration, we denote this particular step by
$\mathcal{C}\longrightarrow_A\mathcal{C}'$. 
An \emph{execution} starting from an initial configuration~$\mathcal{C}_0$
is a sequence $\mathcal{C}_0\longrightarrow_{A_1}\mathcal{C}_1 \longrightarrow_{A_2}\mathcal{C}_2\longrightarrow\dots$ 
of execution steps. Time is measured by the number of steps that have been executed. An execution is completed when it
reaches a configuration in which no node is enabled. After that point, no node is ever activated.
\end{definition}

%The following is easy to prove:
%\begin{proposition}
% If $\mathcal{C}$ is a rooted tree with root~$r$, then after any execution step $\mathcal{C}\longrightarrow_A
%\mathcal{C}'$, $\mathcal{C}'$ is still a directed tree with root~$r$.
%\end{proposition}

If a node was enabled before a particular execution step, was not activated in
that step, and is not enabled after that step, then we say that it was \emph{neutralized} after that step.

\begin{definition}
 Given a particular execution, an \emph{execution round} (or simply \emph{round}) starting from a
configuration~$\mathcal{C}$ consists of the minimum-length sequence of steps in which every enabled node
in~$\mathcal{C}$ is activated or neutralized at least once.
\end{definition}

\begin{remark}
 To simplify the presentation, we will assume throughout the rest of the paper that the arrays containing the IDs of the
children of each node and copies of their height values are consistent with the height values stored by the children
themselves in the current configuration of the system. It should be
clear that maintaining these copies up to date can be achieved with a constant overhead per execution step.
\end{remark}

%% file: Solution.tex
\section{Self-Stabilizing Balancing Solution}
\label{sec:algo}

In this section, we present our self-stabilizing algorithm for balancing containment-based trees, we provide some
termination properties, and we prove that any execution converges to a balanced tree.

% ICDCS

Assuming that each node knows the correct heights of its subtrees, a very simple distributed self-stabilizing algorithm
balances the tree: each node uses the {\it swap} operation whenever two of its children heights are ``too different''.
However, in a distributed context, height information may be inaccurate. This inaccuracy could lead the aforementioned
naive balancing algorithm to make some ``wrong moves.'' For example, in Figure~\ref{fig:edge-swap}, assume that $c$
``thinks'' that its own height is~$4$ and $e$ ``thinks'' that its own height is~$9$ while their actual heights are
respectively~$10$ and $8$; the illustrated {\it swap} would actually unbalance the tree.
% Any node with incorrect height
% information (either overestimating or underestimating its actual height) may cause such moves, questioning algorithm
% terminaison.
On one hand, maintaining heights in a self-stabilizing fashion is easy and ensures that height information
will eventually be correct. On the other hand, no node can know when height information is correct; as a consequence,
height maintenance and balancing have to run concurrently. But their concurrent execution raises an obvious risk: the
{\it swap} operation modifies the tree structure and could thus compromise the convergence of the height maintenance
subprotocol.

Basically, the algorithm that we propose in this section consists of two concurrent actions: one maintaining heights,
the other one balancing the tree.
% Maintenance reads nodes children and modifies nodes height while balancing reads nodes heights and modifies nodes
% children.
We formalize both actions and prove that their concurrent execution converges to a balanced tree.

% ICDCS

\subsection{Algorithm}

At any time~$t$, each node~$u$ is able to evaluate the following functions and predicates for itself and its children:
\begin{itemize}
  \item $\mathrm{max}(x)$: returns any $v\in\children{x}(t)$ such that $h_v(t) = \max_{w \in
\children{x}(t)} h_w(t)$. If $x$ is a leaf, returns~$\bot$ (undefined).
  \item $\mathrm{min}(x)$: returns any $v\in\children{x}(t)$ such that $h_v(t) = \min_{w \in
\children{x}(t)} h_w(t)$. If $x$ is a leaf, returns~$\bot$.
  \item $\mathrm{stable}(x)$: returns true if and only if $h_x(t) = 1 + \max_{w \in \children{x}(t)} h_w(t)$. If $x$ is
a leaf, returns true if and only if $h_x(t) = 0$.
  \item $\mathrm{balanced}(x)$: returns true if and only if for all $z,z'\in\children{x}(t)$,
$|h_z(t)-h_{z'}(t)|\leq\alpha$.
\end{itemize}

When there are more than one possible return values for~$\mathrm{max}(x)$ and~$\mathrm{min}(x)$, an arbitrary choice is
made. For simplicity, we can consider that the candidate node with the smallest ID is returned, although this will not
be crucial for our results.

\newcounter{action}
\newcommand{\GA}[2]{\stepcounter{action}{\bf G\theaction}  #1 & {\bf S\theaction} #2}

Each node~$u$ executes the algorithm in Figure~\ref{alg:algo} (the value of~$\mathrm{stable}(\bot)$ is assumed to be
false).
%
%\vspace{1ex}
\noindent\setcounter{action}{0}
\begin{figure*}
\begin{tabular}{@{}ll@{}} \toprule
{Guard} & {Statement} \\ \midrule
\GA{$\neg\mathrm{stable}(u)$}
{\begin{minipage}[t]{0.5\textwidth}
$h_u \coloneqq 1 + \max_{w \in \children{u}} h_w$ (or $h_u \coloneqq 0$ if $u$ is a leaf)
\end{minipage}} \\[2ex]
\GA{\begin{minipage}[t]{0.42\textwidth}
$\mathrm{stable}(u) \wedge \mathrm{stable}(\mathrm{max}(u)) \wedge$ $\neg\mathrm{balanced}(u)$
\end{minipage}}
{\begin{minipage}[t]{0.5\textwidth}
swap edges~$(u,\mathrm{min}(u))$ and $(\mathrm{max}(u),\mathrm{max}(\mathrm{max}(u))$
\end{minipage}} \\
\bottomrule
\end{tabular}
\caption{\label{alg:algo}Distributed self-stabilizing balancing algorithm}
\end{figure*}
%
%\vspace{1ex}
%
Note that the guards~\textbf{G1} and~\textbf{G2} are mutually exclusive and, by definition, the height update
action effectively has priority over the edge swapping action. We say that a node is \emph{enabled for a height update}
if \textbf{G1} is true, or \emph{enabled for a swap} if \textbf{G2} is true.

When a node~$u$ performs a swap, four nodes are involved: $u$ itself, $\mathrm{max}(u)$, $\mathrm{min}(u)$, and
$\mathrm{max}(\mathrm{max}(u))$. We refer to these nodes as the \emph{source}, the \emph{target}, the \emph{swap-out},
and the \emph{swap-in} nodes of the swap, respectively.

The following proposition can be proved directly from the definitions. We state it without proof.
\begin{proposition}
 If $\mathcal{C}$ is a rooted tree with root~$r$, then after any execution step $\mathcal{C}\longrightarrow_A
\mathcal{C}'$, $\mathcal{C}'$ is still a directed tree with root~$r$.
\end{proposition}

\subsection{Termination Properties}

In the following section, we will prove that for any initial configuration, every possible execution is completed in a
finite number of steps. For the moment, we give two properties of the resulting tree,
assuming of course that the execution consists of a finite number of steps.

The proof of the following proposition can be found in Appendix~\ref{apdx:lem:stablestate}.
\begin{proposition} \label{lem:stablestate}
  If the execution is completed, then in the final configuration~$\mathcal{C}(t^\star)$ all nodes are balanced and have
correct height information.
\end{proposition}

The next proposition, regarding the height of the resulting tree, follows directly from the analysis
in~\cite[Sections~II and~III]{DBLP:journals/tc/LuccioP76}. In our case, the initial conditions of the recurrence studied
in~\cite{DBLP:journals/tc/LuccioP76} are slightly different, but this does not affect the asymptotic behavior of the
height.

\begin{proposition} \label{prop:logheight}
 If the execution is completed in time~$t^\star$, then in the final configuration $h^\star_r(t^\star) = \calO(\log n)$,
where $r$ is the root and $n$ is the number of nodes in the system.
\end{proposition}

%% file: contribution.tex
%\section{Proof of Convergence}
%\label{sec:proof}

\subsection{Proof of Convergence}
\label{subsec:proof}

We give an overview of the proof with references to the appropriate appendices for the technical parts.

The concept of a ``bad node'' will be useful in the analysis of the algorithm.
\begin{definition}[Bad nodes] \label{def:badnodes}
 In a given configuration~$\conf{t}$, an internal node~$u$ is a \emph{bad node} if $h_u(t)\leq
\max_{v\in\children{u}(t)} h_v(t)$. A leaf is a \emph{bad node} if $h_u(t)<0$.
\end{definition}
Intuitively, a bad node is a node that ``wants'' to increase its height value. The proof of the following key lemma can
be found in Appendix~\ref{apdx:cor:nobadnodes}.

\begin{lemma} \label{cor:nobadnodes}
 If $\conf{t}$ contains no bad nodes, then for all $t'\geq t$, $\conf{t'}$ contains no bad nodes.
\end{lemma}

It will be convenient to view any execution of the algorithm as consisting of two
phases: The first phase starts from the initial configuration and ends at the first configuration in which the system is
free of bad nodes. The second phase starts at the end of the first phase and ends at the first configuration in which no
node is enabled, i.e., at the end of the execution. In view of Lemma~\ref{cor:nobadnodes}, the system does not contain
any bad nodes during the second phase. We will prove that each phase is concluded in a finite number of steps, starting
from the second phase.

%\subsection{Second Phase} \label{sec:secondphase}
\subsubsection{Second Phase} \label{sec:secondphase}

We prove convergence for the second phase by bounding directly the number of height updates and the number of swaps that
may occur during that phase. The fact that there are no bad nodes in the second phase is crucial for bounding the number
of height updates. It follows that the number of swaps also has to be bounded, since a long enough sequence of steps in
which only swaps are performed incurs more height updates. The detailed proofs of these claims can be found in
Appendix~\ref{apdx:sec:secondphase}.

\begin{lemma} \label{lem:finiteheightupdates}
 Starting from a configuration with no bad nodes, no execution can perform an infinite number of height updates.
\end{lemma}

\begin{lemma} \label{lem:finiteswaps}
 Starting from a configuration with no bad nodes, no execution can perform an infinite number of swaps.
\end{lemma}

Lemmas~\ref{lem:finiteheightupdates} and~\ref{lem:finiteswaps} directly imply the following:
\begin{theorem}
 In any execution, the second phase is completed in a finite number of steps.
\end{theorem}

%\subsection{First Phase}
\subsubsection{First Phase}

For the sake of presentation,
it will be helpful to sometimes consider that the root of the tree has an imaginary father~$\mathfrak{r}$,
which is never enabled and always a bad node.
\begin{definition}[Extended configuration] \label{def:extconf}
 We denote by~$\confext{t}$ an auxiliary \emph{extended configuration} at time~$t$, which is identical
to~$\conf{t}$ except that the root node in~$\conf{t}$ has a new father node~$\mathfrak{r}$ with
$h_\mathfrak{r}(t)=-\infty$, for all~$t\geq 0$.
\end{definition}
The bad nodes induce a partition of the nodes of the extended configuration into components: each bad node belongs to a
different component, and each non-bad node belongs to the component that contains its nearest bad ancestor.
\begin{definition}[Partition into components] \label{def:components}
 For each bad node~$b$ in~$\confext{t}$, the component~$\mathcal{T}_b(t)$ is the maximal weakly connected
directed subgraph of~$\confext{t}$ that has $b$ at its root and contains no other bad nodes.
\end{definition}

A useful property of this partition is that it remains unchanged as long as the set
of bad nodes remains the same. Therefore, in any sequence of steps in which the set of bad nodes remains the same, each
component behaves similarly to a system that does not contain bad nodes. In particular,
Lemmas~\ref{lem:finiteheightupdates} and~\ref{lem:finiteswaps} imply that each component is stabilized in finite time,
and thus this sequence of steps cannot be infinite. For a complete substantiation of these claims, please refer to the
proof of the following lemma in Appendix~\ref{apdx:lem:componentsstabilize}.
\begin{lemma} \label{lem:componentsstabilize}
 There cannot be an infinite sequence of steps in which no bad node is activated or becomes non-bad.
\end{lemma}

We associate a \emph{badness vector} with each configuration. This vector reflects the distribution of bad nodes in
the system and will serve to quantify a certain notion of progress toward the extinction of bad nodes. In
particular, we will prove that the badness vector decreases lexicographically in every step in which at least one bad
node is activated or becomes non-bad.

\begin{definition}[Badness vector] \label{def:badness}
 For $t\geq 0$, let $b_1,b_2,\dots,b_{|\tilde{B}(t)|}$ be an
ordering of the bad nodes in~$\confext{t}$ by non-decreasing number of bad nodes contained in the
path from~$\mathfrak{r}$ to~$b_i$, breaking ties arbitrarily. Note that $b_1\equiv\mathfrak{r}$.
We define the \emph{badness vector} at time~$t\geq 0$ to be the vector
$$\vec{b}(t)=(|\mathcal{T}_{b_1}(t)|,\dots,|\mathcal{T}_{b_{|\tilde{B}(t)|}}(t)|)\enspace,$$
where the size of a connected component is the number of nodes belonging to that component.
\end{definition}
\begin{figure}[tpb]
\centering
\includegraphics[scale=1]{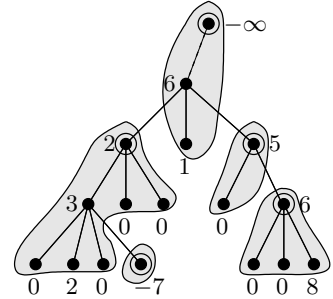}
\caption{An illustration of the notions introduced in
Definitions~\ref{def:badnodes}, \ref{def:extconf}, \ref{def:components}, and~\ref{def:badness}. Node labels indicate
their height values. Circled nodes represent bad nodes. The dashed
edge exists only in the extended configuration and connects the real root of the tree to the artificial
node~$\mathfrak{r}$. Each group of nodes is one component of the partition. The badness vector corresponding to this
configuration is~$(3,7,2,1,4)$.}
\label{fig:components}
\end{figure}
We refer the reader to Figure~\ref{fig:components} for an example of an extended configuration, its partition into
components, and the corresponding badness vector.

\begin{definition}[Lexicographic ordering]
%  Let~$\vec{b}=(x_1,\dots,x_k)$ and~$\vec{b}'=(x'_1,\dots,x'_{k'})$ be two badness vectors.
 Consider two badness vectors $\vec{b}=(x_1,\dots,x_k)$ and~$\vec{b}'=(x'_1,\dots,x'_{k'})$.
 We say that~$\vec{b}'$ is
\emph{lexicographically smaller} than~$\vec{b}$ if one of the following holds:
\begin{enumerate}
 \item $k'<k$, or
 \item $k'=k$ and for some~$i$ in the range~$1\leq i\leq k$, $x'_i<x_i$ and $x'_j=x_j$ for all~$j<i$.
\end{enumerate}
\end{definition}

The proof of the following lemma can be found in Appendix~\ref{apdx:lem:badnesslexdecr}.
\begin{lemma} \label{lem:badnesslexdecr}
%  The simultaneous execution of any set of actions in which at least one bad node is activated or neutralized results
% in a lexicographic decrease of the badness vector.
If at least one bad node is activated or becomes non-bad in step~$\conf{t}\longrightarrow\conf{t+1}$, then
$\vec{b}(t+1)$ is lexicographically smaller than~$\vec{b}(t)$.
\end{lemma}

Lemmas~\ref{lem:componentsstabilize} and~\ref{lem:badnesslexdecr} are the essential ingredients required to prove 
convergence for the first phase. Appendix~\ref{apdx:thm:phaseoneconvergence} contains the full proof.
\begin{theorem} \label{thm:phaseoneconvergence}
 In any execution, the first phase is completed in a finite number of steps.
\end{theorem}

%% file: simulation.tex
\section{Simulation}
\label{sec:simulation}

To investigate the dynamic behavior and properties of our algorithm, we implemented a round based simulator.
Each simulation ($i$) builds a full binary tree, ($ii$) initializes heights, and ($iii$) runs the balancing protocol.
We used a synchronous daemon to run simulations, i.e., all the enabled nodes of the system execute their enabled actions
simultaneously.
%While no synchrony hypothesis have been made in the proof, we could have used weaker daemon.
The impact of non-deterministic daemons and/or the weaker daemon that is required to run our algorithm will be tackled
in future work.

In the following, for a given simulation we will denote by~$n$ the number nodes of the tree, by~$h_i$ its initial height
(i.e.\ after generation), by~$h_f$ its final height (i.e.\ after balancing) and by~$t$ the execution time in rounds.
%As our algorithm converges and as we are working on binary trees $h_f = \calO(\log n)$.

\subsection{Almost Linear Trees}
\label{linear-tree-case}

Intuitively, almost linear trees are stressing for a balancing algorithm because they are ``as unbalanced as possible''.
In the following, %, trees are thus generated to be ``as linear as possible'' for binary full trees.
%Basically one non leaf node has two leaves as children and all other non leaf nodes have exactly one leaf and one non leaf node as children.
%As trees are full and binary, all $n$ nodes instances of height $n/2$ are structurally equivalent.
for a given~$n$, the initial tree is the structurally unique full binary tree of height $\lceil n/2\rceil$.
The only unspecified part of the simulation is the initial height values.
It turns out that this has practically a very small impact on termination time.
For each $n$ we ran thousands of simulations starting from the corresponding linear tree.
For a given $n$ they always converged in the same number of rounds.
As a consequence, all runs starting from the same linear tree will have approximately the same results.

\begin{figure}[tpb]
	\centering
	\begin{minipage}{.48\textwidth}
		\centering
		\includegraphics[angle=-90,scale=0.45]{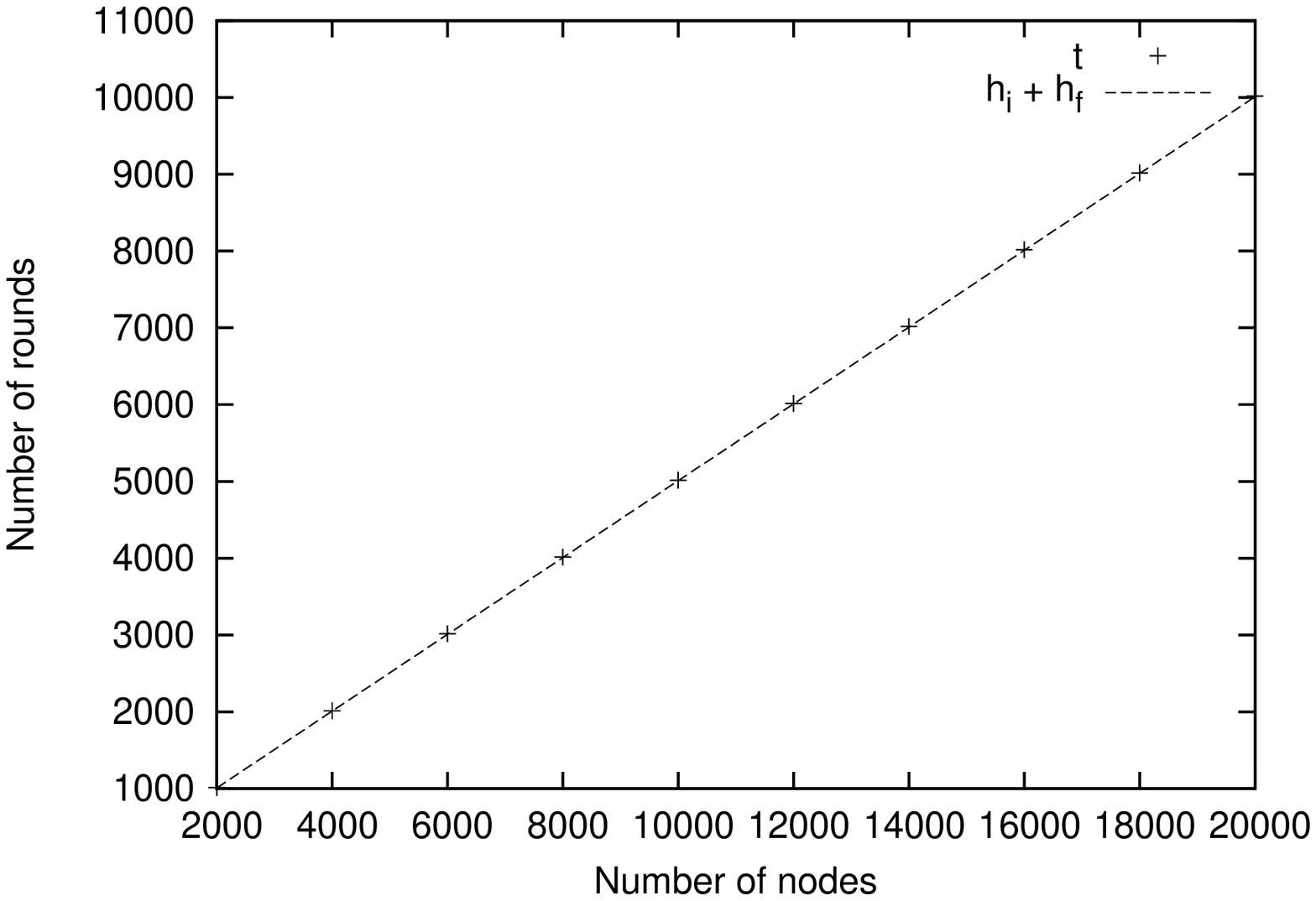}
		\caption{\label{fig:round-linear}$t$ for almost linear trees}	
	\end{minipage}
	\begin{minipage}{.48\textwidth}
		\centering
		\includegraphics[angle=-90,scale=0.45]{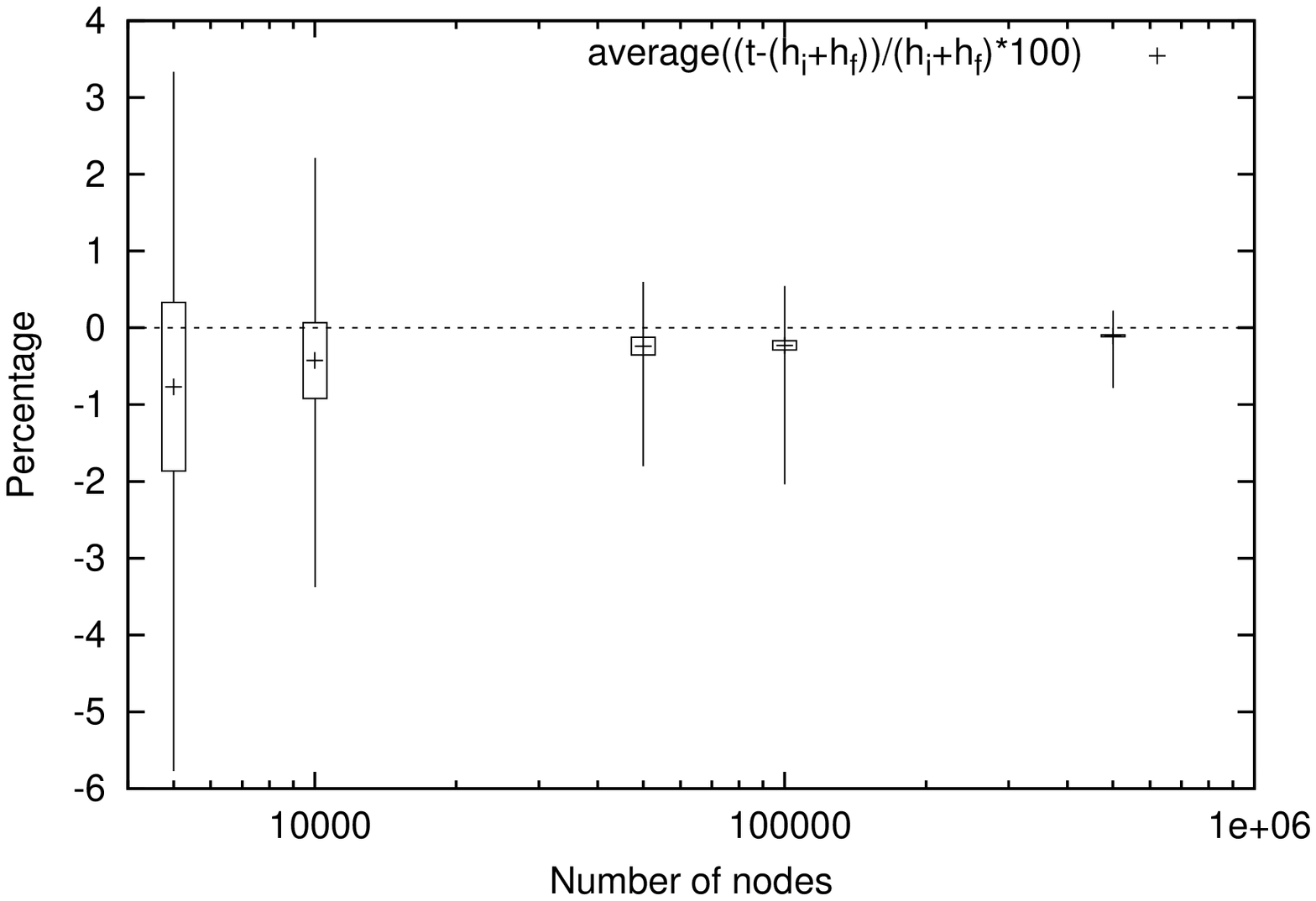}
		\caption{\label{fig:round-random}$t$ for random trees}
	\end{minipage}
\end{figure}

Figure~\ref{fig:round-linear} shows the termination time of simulations for different numbers of nodes.
% The X-axis represents the number of nodes while the Y-axis is the number of rounds.
It contains two curves: the first one is plotted from our experiment and
each point stands for an almost linear tree.
The second curve is the sum of initial and final tree heights.
%As the tree is full and binary this curve is $n/2+\log_2(n)$.

This plot tends to show that the round complexity of our algorithm is $\calO(n)$ in the worst case.
%More precisely we expect its average complexity to be $\calO(h_i+h_f)$ which also holds for the linear case as $n/2>>\log_2(n)$.

\subsection{Random Trees}

To showcase the applicability of our algorithm to an existing system,
% and to check if $\calO(h_i+h_f)$ is a practicaly viable approximation of its time complexity,
trees are generated using the join protocol of~\cite{PexpanderVtree}.

Figure~\ref{fig:round-random} shows the distance between the sum of initial and final tree heights and experimental
termination times.
% X-axis stands for the number of nodes.
The vertical axis gives the average variation between $h_i+h_f$ and experimental termination times. 
%The experimental is the measured one while the theoretical ---borrowed from~\cite{DLPT} and linear tree case~\ref{linear-tree-case}--- is the sum of initial and final tree heights.
For each $n$ we ran thousands of simulations.
Each candlestick sums up statistics on those runs; the whiskers indicates minimum and maximum variation, the
cross indicates the average variation and the box height indicates the standard deviation. 

%The standard deviation decreases with $n$ as a consequence of tree height distribution~\cite{GOR93}.
The greater $n$ is, the closer $h_i+h_f$ and experimental results are.
This result indicates that the average round complexity of our algorithm is $\calO(h_i+h_f)$.

%% file: conclusion.tex
\section{Concluding Remarks}
\label{sec:conclusion}

In this work, we propose a new distributed self-stabilizing algorithm to rebalance containment-based trees
using edge swapping. Simulation results indicate that the algorithm is quite efficient in terms of round complexity; in
fact, it seems that we can reasonably expect $\calO(n)$ to be a worst-case bound, whereas in the average case the
running time is closer to $\calO(h_i+h_f)$ rounds. Interestingly, this average-case bound also appears
in a different setting in~\cite{DLPT}. Note that the conjectured average-case bound is close to $\calO(\log n)$ in a
practically relevant scenario in which some faults appear (or new nodes are inserted) in an already balanced tree.
%We have shown that this primitive is sufficient to balance any containment based tree in a distributed way.
%Moreover, our solution is silent. %The initial configuration is any arbitrary configuration where node's variables
%could be corrupted.
% Borrowed from our simulations results and~\cite{DLPT}, we conjuncture that its average round
% complexity is $\calO(h_i+h_f)$.

We have assumed that nodes keep correct copies of the height values of their children, so that each node can read the
height values of its grandchildren by looking at the memory of its children. For simplicity, we have not dealt with the
extra synchronization that would be required to maintain these copies up-to-date, but it should be possible to achieve
this with a constant overhead per execution step. Furthermore, we have assumed that internal nodes have degree at
least~two. Degenerate internal nodes with degree~one could be accommodated by a bottom-up protocol that runs in parallel
and essentially disconnects them from the tree, attaching their children to their parents. Finally, note that in
Section~\ref{sec:primitive} we remarked that edge swaps may have semantic impact if they rearrange nodes so as to
violate the containment relation. This can also be fixed by another bottom-up protocol that restores each node's label
to the minimum that suffices to contain the labels of its children.

Possible directions for future work include establishing the conjectured upper bounds of $\calO(n)$ and
$\calO(h_i+h_f)$ for the round complexity.
We already have some preliminary results in this direction: the first phase
of the algorithm (refer to %Section~\ref{sec:proof}) 
Subsection~\ref{subsec:proof}) is indeed concluded in~$\calO(n)$ rounds.
An extension of this work would be to adapt the proposed algorithm in the message passing model.

%\rabbit

%% file: appendix.tex
\subsection{Proof of Proposition~\ref{lem:stablestate}} \label{apdx:lem:stablestate}

% \begin{proof}[Proposition~\ref{lem:stablestate}]
 We prove that each node~$u$ is balanced and has correct height information by induction on the actual height of~$u$
in~$\conf{t^\star}$.

If $h^\star_u(t^\star)=0$, then $u$ is a leaf and since it is not enabled, $h_u(t^\star)=0$. Therefore, $u$ has correct
height information and is trivially balanced. Assume that all nodes of actual height~$k$ or less, where $k\geq 0$, are
balanced and have correct height information in $\conf{t^\star}$. Consider a node~$u$ at actual height~$k+1$. Since $u$
is not enabled for a height update, $h_u(t^\star)=\max_{v\in\children{u}(t^\star)} h_v(t^\star) +1$. But all children
of~$u$ are at actual height~$k$ or less, therefore the inductive hypothesis and the last equality imply:
$h_u(t^\star)= \max_{v\in\children{u}(t^\star)} h^\star_v(t^\star) +1 = h^\star_u(t^\star)$. Moreover, $u$ is not
enabled for a swap, which means that for all children~$v$ of~$u$, $h_v(t^\star)\geq h_u(t^\star)-1-\alpha$, and because
they have correct height information, $h^\star_v(t^\star)\geq h^\star_u(t^\star)-1-\alpha$. Therefore, $u$ has correct
height information and is balanced.
% \qed
% \end{proof}

\subsection{Proof of Lemma~\ref{cor:nobadnodes}} \label{apdx:cor:nobadnodes}

We will use the following notation in this section and in Appendices~\ref{apdx:lem:componentsstabilize}
and~\ref{apdx:lem:badnesslexdecr}. Let $P=\mathcal{C}\longrightarrow_A\mathcal{C}'$ be a step of the execution. The
set~$A$
is partitioned into subsets~$A_h$, $A_b$, and~$A_s$, where $A_h$ is the set of non-bad nodes which perform a height
update, $A_b$ is the set of bad nodes which perform a height update, and $A_s$ is the set of nodes which are sources of
a swap. Let $A_t$ be the set of nodes which are targets of a swap and let~$B$
and~$B'$ denote the set of bad nodes in configuration~$\mathcal{C}$ and~$\mathcal{C}'$, respectively. For each
node~$u$, let~$\children{u}$ be the set of children of~$u$ in~$\mathcal{C}$ and~$\subtree{u}$ be the set of nodes of
the subtree rooted at~$u$ in~$\mathcal{C}$, and let~$\children{u}'$ and~${\subtree{u}}'$ be the corresponding sets 
in~$\mathcal{C}'$. Finally, let $\mathcal{G}(A_b)$ be the subgraph induced by~$A_b$ in the
configuration~$\mathcal{C}$.

The following lemma states some very basic properties that are easily derived from the definitions. It will be used
implicitly throughout the proofs. We state it without proof.
\begin{lemma}[Easy properties]
 \begin{enumerate}
  \item $A_b\subseteq B$.
  \item $A_s\cap A_h=A_s\cap B=A_t\cap A_h=A_t\cap B=\emptyset$.
  \item For all nodes~$u$, $\children{u}'=\children{u}$ if and only if $u\not\in A_s\cup A_t$.
  \item The set of leaves in~$\mathcal{C}'$ is equal to the set of leaves in~$\mathcal{C}$.
 \end{enumerate}
\end{lemma}

\begin{lemma} \label{lem:badleaves}
 Let~$\mathcal{E}$ be a weakly connected component of~$\mathcal{G}(A_b)$. In~$\mathcal{C}'$, the nodes of~$\mathcal{E}$
still induce a weakly connected subgraph which is identical to~$\mathcal{E}$. No leaf of~$\mathcal{E}$ belongs to $B'$.
\end{lemma}
\begin{proof}
% \begin{proof}[\ref{lem:badleaves}]
No node of~$\mathcal{E}$ is in~$A_s\cup A_t$, therefore $\children{v}'=\children{v}$ for all nodes~$v$
of~$\mathcal{E}$. This suffices to prove that the nodes of~$\mathcal{E}$ induce in~$\mathcal{C}'$ a weakly connected
subgraph that is identical to~$\mathcal{E}$.

Now, let~$u$ be a leaf of~$\mathcal{E}$. If $u$ is also a leaf in~$\mathcal{C}$, then its activation has set its height
variable to~$0$ in~$\mathcal{C}'$. But $u$ is still a leaf in~$\mathcal{C}'$, therefore it now has the correct
height value and therefore $u\not\in B'$.
If $u$ is not a leaf of~$\mathcal{C}$, it means that
$\children{u}$  is non-empty and
$\children{u}\cap A_b=\emptyset$ (if $v\in\children{u}\cap A_b$, then $v$ would also be in~$\mathcal{E}$ and $u$ would
not
be a leaf of~$\mathcal{E}$). We deduce that each node in~$\children{u}$ either
decreases or retains its height
variable in~$\mathcal{C}'$. Moreover, by definition of the height update action, node~$u$ adjusts its height variable to
at least one greater than any of the values of the height variables of its children in~$\mathcal{C}$. We conclude that
$u\not\in B'$. %\qed
\end{proof}

\begin{lemma} \label{lem:nonact}
%  If a node is non-activated and non-bad in~$\mathcal{C}$, then in~$\mathcal{C}'$ it is non-bad, unless
% it is the father of one or more nodes in the set~$A_b$.
If $u\not\in A\cup B$ and $\children{u}'\cap A_b=\emptyset$, then $u\not\in B'$.
\end{lemma}
\begin{proof}
% \begin{proof}[\ref{lem:nonact}]
Since $u$ is not activated, the value of its height variable remains the same in~$\mathcal{C}'$. 
Assume, first, that~$u\not\in A_t$. Then, $\children{u}'=\children{u}$ (since $u\not\in A_s$, either).
Moreover, $\children{u}'\cap A_b=\emptyset$, therefore
all nodes in~$\children{u}=\children{u}'$ either decreased or retained their height variables. From these observations
and the fact that $u\not\in B$, we conclude that $u\not\in B'$.

Now, assume that $u\in A_t$, and let~$x\not\in A_b$ be its new child in~$\mathcal{C}'$. By the definition of the swap
guard, the original height variable of node~$x$ is strictly
smaller than the height variable of~$u$, and it may have decreased even further if $x$ was simultaneously activated for
a height update. All other chilren of~$u$ in~$\mathcal{C}'$ were also children of~$u$ in~$\mathcal{C}$, and, by the
fact that $\children{u}'\cap A_b=\emptyset$, we know that they either decreased or retained their height values.
From these observations and the fact that $u\not\in B$, we have that $u\not\in B'$. %\qed
\end{proof}

\begin{lemma} \label{lem:nonbad}
%  If a node is non-bad in~$\mathcal{C}$ and performs a height update, then in~$\mathcal{C}'$ it remains non-bad, unless
% it is the father of one or more nodes in the set~$A_b$.
If $u\in A_h$ and $\children{u}'\cap A_b=\emptyset$, then $u\not\in B'$.
\end{lemma}
\begin{proof}
% \begin{proof}[\ref{lem:nonbad}]
Since $u\in A_h$, we know that $u\not\in A_s\cup A_t$ and therefore $\children{u}'=\children{u}$. Since
$\children{u}'\cap A_b=\emptyset$, all its children either retain or decrease their height variable, whereas the height
variable of~$u$ is adjusted to at least one greater than any of the original values of the height variables of its
children. Thus, $u\not\in B'$.  %\qed
\end{proof}

\begin{lemma} \label{lem:swap}
%  If a node performs a swap, then in~$\mathcal{C}'$ it remains non-bad, unless it is the father of one or more nodes in
% the set~$A_b$.
If $u\in A_s$ and $\children{u}'\cap A_b=\emptyset$, then $u\not\in B'$.
\end{lemma}
\begin{proof}
% \begin{proof}[\ref{lem:swap}]
 Let~$x\not\in A_b$ be the new child of~$u$ in~$\mathcal{C}'$. By definition of the
swap guard, we know that the height variable of~$x$ in~$\mathcal{C}'$ is strictly smaller than that of~$u$. All
other children of~$u$ were also children of~$u$ in~$\mathcal{C}$ and they have either retained or decreased their height
variable, therefore $u\not\in B'$.  %\qed
\end{proof}

\begin{lemma} \label{lem:newbadnodes}
 $|B'\setminus B|\leq |A_b\setminus B'|$.
\end{lemma}
\begin{proof}
% \begin{proof}[\ref{lem:newbadnodes}]
To each node~$x\in B'\setminus B$, we associate a node~$y\in A_b\setminus B'$ as follows:
By Lemmas~\ref{lem:nonact},~\ref{lem:nonbad}, and~\ref{lem:swap}, a non-bad node~$x$ in~$\mathcal{C}$ can be turned
into a bad node in~$\mathcal{C}'$ only if, in $\mathcal{C}'$, it is the father of some node~$b\in A_b$. In fact, this
node~$b$
must be the root of some weakly connected component~$\mathcal{E}$ of~$\mathcal{G}(A_b)$: Otherwise, $b$ would have a
father $u\in A_b$ in~$\mathcal{C}$, and by Lemma~\ref{lem:badleaves}, $u$ would still be the father of~$b$
in~$\mathcal{C}'$, thus $u\equiv x$. This contradicts with the fact that $x\not\in B$.
We define~$y$ to be any leaf of~$\mathcal{E}$, which, by Lemma~\ref{lem:badleaves}, became non-bad in~$\mathcal{C}'$.

To conclude the argument, note that two distinct nodes $x,x'\in B'\setminus B$ must be the fathers of the roots of two
distinct components of~$\mathcal{G}(A_b)$, and therefore they are associated to two distinct nodes~$y$ and~$y'$. %\qed
\end{proof}

\begin{lemma} \label{lem:badnodenotincr}
% The simultaneous execution of any set of actions cannot increase the number of bad nodes.
$|B'|\leq |B|$.
\end{lemma}
\begin{proof}
% \begin{proof}[\ref{lem:badnodenotincr}]
One can easily verify that, for any sets~$B$ and~$B'$, $|B'|-|B|=|B'\setminus
B|-|B\setminus B'|$. Since $A_b\subseteq B$, we have that $|A_b\setminus B'|\leq|B\setminus B'|$, therefore
$|B'|-|B| \leq |B'\setminus B|-|A_b\setminus B'|$. This, combined with Lemma~\ref{lem:newbadnodes}, yields $|B'|\leq
|B|$. %\qed
\end{proof}

Lemma~\ref{cor:nobadnodes} follows immediately from Lemma~\ref{lem:badnodenotincr}.

\subsection{Missing Proofs from Subsection~\ref{sec:secondphase}} \label{apdx:sec:secondphase}

\begin{lemma} \label{lem:creationbadnode}
 For any node~$u$, if $h_u(t)<h^\star_u(t)$, then there exists in~$\conf{t}$ at least one bad node in the subtree rooted
at~$u$.
\end{lemma}
\begin{proof}
%\begin{proof}[\ref{lem:creationbadnode}]
 By induction on the actual height of~$u$. If $h^\star_u(t)=0$, then $u$ is a leaf and, by assumption,
$h_u(t)<h^\star_u(t)=0$.
Therefore, $u$ is a bad node. Now, assume that the statement holds for all nodes with actual height at most~$k$,
where~$k\geq 0$. Consider a node~$u$ with~$h^\star_u(t)=k+1$ and let $v$ be one of its children with actual
height~$h^\star_v(t)=k$ (at least one such child must exist). We can assume that~$u$ is not bad, otherwise the claim is
proved. Since $u$ is not bad, $h_v(t)\leq h_u(t)-1$. By assumption, $h_u(t)\leq h^\star_u(t)-1$, thus we get
$h_v(t)\leq h^\star_u(t)-2=k-1$. Since the actual height of~$v$ is~$k$, we have $h_v(t)<h^\star_v(t)$ and, by the
inductive hypothesis, there exists a bad node in the subtree rooted at~$v$. %\qed
\end{proof}

\begin{lemma} \label{lem:heightupdateenabledpersists}
 In the second phase, if a node becomes enabled for a height update, it will remain enabled for a height update at least
until it is activated.
\end{lemma}
\begin{proof}
% \begin{proof}[\ref{lem:heightupdateenabledpersists}]
 Note that a node that is enabled for a height update cannot be the source or the target of a swap, therefore its set
of children does not change while it is enabled for a height update. Moreover, its children cannot increase their own
height values, so the node will remain enabled for a height update at least until its activation. %\qed
\end{proof}

\begin{proof}[Proof of Lemma~\ref{lem:finiteheightupdates}]
Consider an execution of the algorithm in which an infinite number of height updates are executed. Since the number of
nodes is finite, at least one node must execute a height update an infinite number of times. By the fact that the
initial configuration contains no bad nodes and by Lemma~\ref{cor:nobadnodes}, each time that node executes a
height update, its height variable decreases. At some point, its height variable will become negative and at that
point, by Lemma~\ref{lem:creationbadnode}, some node in its subtree will become bad. This contradicts with
Lemma~\ref{cor:nobadnodes}. %\qed
\end{proof}

\begin{proof}[Proof of Lemma~\ref{lem:finiteswaps}]
By Lemma~\ref{lem:finiteheightupdates}, there exists a finite time~$t_0$ after which no height updates are performed.
For each
node~$u$, let $h_u$ denote the value of its height variable at time~$t_0$ and, since it remains constant thereafter, at
all subsequent times. Furthermore, for $t\geq t_0$, let $\maxheightchildren{u}(t)$ denote the set of children of~$u$
at time~$t$ whose height variable is equal to $h_u-1$.

Note that $u$ is enabled for a height update at time~$t$ if and
only if $|\maxheightchildren{u}(t)|=0$.
We observe now that in every step, if $u$ is the target of a swap then $|\maxheightchildren{u}(t)|$ is decreased
by~$1$, otherwise it remains the same. By Lemma~\ref{lem:heightupdateenabledpersists}, if $|\maxheightchildren{u}(t)|$
becomes~$0$ then it remains equal to~$0$ until $u$ performs a height update.

Suppose, now, for the sake of contradiction, that an infinite number of swaps are performed after time~$t_0$. For each
swap, there exists a node that is the target of that swap. It follows, then, that after at most $\sum_{u}
|\maxheightchildren{u}(t_0)|$ swaps have been performed after time~$t_0$, all nodes in the system will be either idle
or enabled for a height update (idle nodes will include nodes that are so low in the tree that they cannot possibly be
the source of a swap and the root of the tree which will not be able to perform a swap since all of its children will
be enabled for a height update). At that point, either all nodes are idle and thus the execution is completed, which
contradicts with the fact that an infinite number of swaps are performed after time~$t_0$, or the only choice of the
scheduler is to activate a node for a height update, which contradicts with the fact that no height updates are
performed after time~$t_0$. %\qed
\end{proof}

\subsection{Proof of Lemma~\ref{lem:componentsstabilize}} \label{apdx:lem:componentsstabilize}

In this section, we use some notation introduced in Appendix~\ref{apdx:cor:nobadnodes}. Additionally, in this section
and in Appendix~\ref{apdx:lem:badnesslexdecr}, we will use the following notation.
Let~$\tilde{\mathcal{C}}$ and~$\tilde{\mathcal{C}}'$ denote the extended configurations
corresponding to~$\mathcal{C}$ and~$\mathcal{C}'$, and let~$\tilde{B}$ and~$\tilde{B}'$ denote the corresponding sets
of bad nodes. Moreover, let~$\{\mathcal{T}_b\}_{b\in B}$ and~$\{\mathcal{T}'_b\}_{b\in B'}$ be the partition of nodes
into components for the two configurations, and if~$\mathcal{T}$ is any such component, let~$V(\mathcal{T})$ denote its
set of nodes.

\begin{definition}[Swap chain]
 A \emph{swap chain} in step~$P$ is a maximal-length directed path $u_0,\dots,u_\sigma$ ($\sigma\geq 1$) in
configuration~$\mathcal{C}$ such that $u_0,\dots,u_{\sigma-1} \in A_s$, $u_1,\dots,u_\sigma\in A_t$, and, if
$\sigma\geq 2$, $u_2,\dots,u_\sigma$ are the swap-ins of the swaps performed by nodes~$u_0,\dots,u_{\sigma-2}$,
respectively.
\end{definition}

\begin{lemma}[Properties of swap chains] \label{lem:swapchainprop}
 Let $u_0,\dots,u_\sigma$ be a swap chain in step~$P$.
\begin{enumerate}
 \item \label{lem:swapchainprop:1} $u_1\in\children{u_0}'$.
 \item \label{lem:swapchainprop:2} In~$\mathcal{C}'$, the nodes of even order in the swap chain ($u_0,u_2,\dots$) induce
a directed path starting
from~$u_0$. Similarly, the nodes of odd order in the swap chain ($u_1,u_3,\dots$) induce a directed path starting
from~$u_1$.
 \item \label{lem:swapchainprop:3} $\bigcup_{i=0}^\sigma \children{u_i}' = \bigcup_{i=0}^\sigma \children{u_i}$.
\end{enumerate}
\end{lemma}
\begin{proof}
% \begin{proof}[\ref{lem:swapchainprop}]
 Property~\ref{lem:swapchainprop:1} follows immediately from the fact that $u_0$ is the first node in the swap chain,
therefore it is not the target of any swap operation. Property~\ref{lem:swapchainprop:2} follows from the fact that
$u_2,u_4,\dots$ are the swap-in nodes for the swaps performed by nodes $u_0,u_2,\dots$ respectively and, similarly,
$u_3,u_5,\dots$ are the swap-in nodes for the swaps performed by nodes $u_1,u_3,\dots$ respectively. For
Property~\ref{lem:swapchainprop:3}, note that the nodes in the swap chain exchange children only with other nodes in
the swap chain. %\qed
\end{proof}

\begin{lemma} \label{lem:subtreepreserve}
 If $u\not\in A_t$, then ${\subtree{u}}'=\subtree{u}$.
\end{lemma}
\begin{proof}
% \begin{proof}[\ref{lem:subtreepreserve}]
 By induction on $|{\subtree{u}}'|$. If $|{\subtree{u}}'|=1$, then $u$ is a leaf in~$\mathcal{C}'$, thus also
in~$\mathcal{C}$, and ${\subtree{u}}'=\subtree{u}=\{u\}$. Assume that the statement holds for all nodes whose subtree
contains at most~$k$ nodes in~$\mathcal{C}'$, where $k\geq 1$. Let $u\not\in A_t$ be a node for which
$|{\subtree{u}}'|=k+1$.

If $u\not\in A_s$, then $\children{u}'=\children{u}$. Moreover, for all~$z\in\children{u}'$,
$|{\subtree{z}}'|<|{\subtree{u}}'|$ and $z\not\in A_t$, since their father was not the source of a swap. Therefore, the
inductive hypothesis applies to each node~$z\in\children{u}'$ and we get:
$${\subtree{u}}'=\{u\}\cup\bigcup_{z\in\children{u}'} {\subtree{z}}' = \{u\}\cup\bigcup_{z\in\children{u}}
\subtree{z} = \subtree{u} \enspace.$$

If $u\in A_s$, then $u$ must be the origin of a swap chain. Let $V_\mathrm{ch}$ be the node set of that swap chain and
let $\children{\mathrm{ch}}=\left(\bigcup_{v\in V_\mathrm{ch}} \children{v}\right)\setminus V_\mathrm{ch}$ and
$\children{\mathrm{ch}}'=\left(\bigcup_{v\in V_\mathrm{ch}} \children{v}'\right)\setminus V_\mathrm{ch}$. By
Lemma~\ref{lem:swapchainprop}, $\children{\mathrm{ch}}'=\children{\mathrm{ch}}$. Moreover, each node
$z\in\children{\mathrm{ch}}'$ is in the subtree of~$u$ in~$\mathcal{C}'$ and 
$z\not\in A_t$. Therefore, the inductive hypothesis applies to each node $z\in\children{\mathrm{ch}}'$ and we get:
$${\subtree{u}}'= V_\mathrm{ch}\cup\bigcup_{z\in \children{\mathrm{ch}}'} {\subtree{z}}' =
V_\mathrm{ch}\cup\bigcup_{z\in \children{\mathrm{ch}}} \subtree{z} = \subtree{u} \enspace.$$ %\qed
\end{proof}

\begin{lemma} \label{lem:samebadset}
 If $\tilde{B}'=\tilde{B}$, then $V(\mathcal{T}'_b)=V(\mathcal{T}_b)$, for all~$b\in \tilde{B}$.
\end{lemma}
\begin{proof}
% \begin{proof}[\ref{lem:samebadset}]
Note that, for each $b\in \tilde{B}'$, we can write $V(\mathcal{T}'_b)$ as follows:
$$V(\mathcal{T}'_b) = {\subtree{b}}' \setminus \bigcup_{u\in
\tilde{B}'\cap {\subtree{b}}' \setminus \{b\}} {\subtree{u}}' \enspace.$$
Since $b\in\tilde{B}'=\tilde{B}$, we must have $b\not\in A_t$. Therefore, by Lemma~\ref{lem:subtreepreserve}, each
$b\in\tilde{B}'$ satisfies ${\subtree{b}}'=\subtree{b}$. We get, then, that
$$V(\mathcal{T}'_b) = {\subtree{b}} \setminus
\bigcup_{u\in \tilde{B}\cap {\subtree{b}} \setminus \{b\}}
{\subtree{u}} = V(\mathcal{T}_b) \enspace.$$
%\qed
\end{proof}

% \begin{proof}[of Lemma~\ref{lem:componentsstabilize}]
 By Lemma~\ref{lem:samebadset}, in a sequence of steps in which no bad node is activated or becomes non-bad, each
of the connected components behaves in the same way as a tree that does not contain bad nodes.\footnote{That is
slightly inaccurate: Certain nodes of the component may contain in their children set some bad nodes, which are at the
root of other components. From the point of view of the father's component, these bad nodes will behave as if they are
leaves whose height value is fixed to some arbitrary value, smaller than the height value of their father.
However, it should be clear that this does not change the fact that the component will stabilize after a finite number
of steps.} Therefore, by
Lemmas~\ref{lem:finiteheightupdates} and~\ref{lem:finiteswaps}, each component will stabilize in finite time and the bad
nodes will be the only candidates for activation. Lemma~\ref{lem:componentsstabilize} follows.
%  \qed
% \end{proof}

\subsection{Proof of Lemma~\ref{lem:badnesslexdecr}} \label{apdx:lem:badnesslexdecr}

In this section, we use some notation introduced in Appendices~\ref{apdx:cor:nobadnodes}
and~\ref{apdx:lem:componentsstabilize}. Additionally, in this section we will use the following notation.
Let~$\vec{b}$ and~$\vec{b}'$ denote the badness vectors corresponding to~$\mathcal{C}$
and~$\mathcal{C}'$.

\begin{lemma} \label{lem:badnodeneutralized}
% In any step in which a bad node is neutralized, the number of bad nodes decreases.
If at least one bad node becomes non-bad without being activated in step~$P$, then $|B'|<|B|$.
\end{lemma}
\begin{proof}
% \begin{proof}[\ref{lem:badnodeneutralized}]
 Let $N$ be the set of bad nodes that become non-bad without being activated. We can partition the set~$B$ as follows:
$$B=N\cup \left(A_b\setminus B'\right) \cup \left(B\cap B'\right)
\enspace.$$
Moreover, we can naturally partition the set~$B'$ as follows:
$$B'=\left(B'\setminus B\right) \cup \left(B\cap B'\right) \enspace.$$
If $|N|>0$, then from the first equation we get $|B|>|A_b\setminus B'|+|B\cap B'|$, and then from the second equation
and
Lemma~\ref{lem:newbadnodes} we have that $|B'|=|B'\setminus B|+|B\cap B'|\leq |A_b\setminus B'|+|B\cap B'|<|B|$. %\qed
\end{proof}

\begin{lemma} \label{lem:componentfatherbad}
%  In any step in which a bad node is activated but its father is a bad node that is not activated, the number of bad
% nodes decreases.
If there exist nodes~$u$ and~$v$ such that $v\in\children{u}'\cap A_b$ and $u\in B\setminus A$ or $u\not\in B'$, then
$|B'|<|B|$.
\end{lemma}
\begin{proof}
% \begin{proof}[\ref{lem:componentfatherbad}]
 We will prove the statement by demonstrating an injection from~$B'$
to~$B\setminus\{v\}$. Consider the function $f:B'\rightarrow B$ where, given~$x\in B'$, $f(x)$ is defined as follows:
\begin{itemize}
 \item If $x\in B\setminus A_b$, then $f(x)=x$.
 \item Otherwise, $f(x)=y$ where $y$ is any child of~$x$ in~$\mathcal{C}'$ such that $y\in A_b$.
\end{itemize}
We need to prove that the function~$f$ is well-defined and injective. For injectivity, it suffices to show that if
$x\not\in B\setminus A_b$, then $f(x)\not\in B\setminus A_b$. Indeed, given an~$x\in B'$ such that $x\not\in
B\setminus A_b$, we distinguish two cases: If $x\not\in B$, then, by Lemmas~\ref{lem:nonact},~\ref{lem:nonbad},
and~\ref{lem:swap}, $x$ has a child~$y\in A_b$ in~$\mathcal{C}'$. On the other
hand, if $x\in A_b$, then by Lemma~\ref{lem:badleaves} we know that $x$ is not a leaf of the component of
$\mathcal{G}(A_b)$ in which it belongs and thus it has a child~$y\in A_b$ in~$\mathcal{C}'$. Clearly, in both cases,
$y\not\in B\setminus A_b$. 

It remains to show that no $x\in B'$ is mapped to~$v$. The only candidate nodes that could be mapped to~$v$ are $v$
itself and~$u$, the father of~$v$ in~$\mathcal{C}'$. We know that $v\in A_b$, which implies that $v\not\in B\setminus
A_b$, and thus $f(v)\neq v$. As for~$u$, we have two cases: If $u\not\in B'$, then $u$ is not even in the domain of~$f$.
If $u\in B'$, then, by assumption, we must have $u\in B\setminus A$ and thus $f(u)=u\neq v$. %\qed
\end{proof}

% \begin{proof}[of Lemma~\ref{lem:badnesslexdecr}]
For the proof of Lemma~\ref{lem:badnesslexdecr}, 
we can assume that $\vec{b}'$ and $\vec{b}$ are of equal length, otherwise the statement holds by
Lemma~\ref{lem:badnodenotincr}. Moreover, we can assume that $A_b\neq\emptyset$, otherwise the statement holds by
Lemma~\ref{lem:badnodeneutralized}. Since $\mathfrak{r}\not\in A_b$, there exists a bad node in~$\tilde{\mathcal{C}}$
whose corresponding component contains the parent of a node in~$A_b$. Let $b_j$, $j\geq 1$, be the first such bad node
in the ordering of Definition~\ref{def:badness}.

 By definition of~$b_j$, we have that $b_1,\dots,b_j\not\in A_b$. Therefore, by Lemma~\ref{lem:badnodeneutralized},
$b_1,\dots,b_j\in \tilde{B}'$. By Lemmas~\ref{lem:nonact},~\ref{lem:nonbad}, and~\ref{lem:swap}, a non-bad node can be
turned into a bad node only if, in $\tilde{\mathcal{C}}'$, it is the father of a node in~$A_b$. This implies that
$b_1,\dots,b_j$ are the first $j$ bad nodes in~$\tilde{\mathcal{C}}'$ according to the ordering of
Definition~\ref{def:badness}.

 By definition of~$b_j$ and Lemma~\ref{lem:badnodeneutralized}, we also know that any child of any node in any of the
components $\mathcal{T}_{b_1},\dots,\mathcal{T}_{b_{j-1}}$, that was in~$\tilde{B}$, is also in~$\tilde{B}'$. Therefore,
for each $i<j$, $|\mathcal{T}'_{b_i}|=|\mathcal{T}_{b_i}|$. Finally, by Lemma~\ref{lem:componentfatherbad}, we can
assume that $b_j$ itself is not the father of any node $v\in A_b$ and that, for
any such node~$v$ whose father~$w$ was in~$\mathcal{T}_{b_j}$, at least~$w$ no longer belongs to
$\mathcal{T}'_{b_j}$. Therefore, $|\mathcal{T}'_{b_j}|<|\mathcal{T}_{b_j}|$.
%  \qed
% \end{proof}

\subsection{Proof of Theorem~\ref{thm:phaseoneconvergence}} \label{apdx:thm:phaseoneconvergence}

%\begin{proof}
% \begin{proof}[Theorem~\ref{thm:phaseoneconvergence}]
 By Lemma~\ref{lem:badnesslexdecr}, the badness vector decreases lexicographically whenever at least one bad node is
activated or becomes non-bad. By Lemma~\ref{lem:componentsstabilize}, we cannot have an infinite sequence of steps in
which no bad node is activated or becomes non-bad. Moreover, during such a sequence of steps, the set of bad nodes
remains the same by Lemmas~\ref{lem:nonact}, \ref{lem:nonbad}, and~\ref{lem:swap}, and thus the badness vector remains
the same by Lemma~\ref{lem:samebadset}. The theorem follows from these observations and the easy fact that no
configuration can have a corresponding badness vector that is lexicographically smaller than the single-component
badness vector $(n+1)$, where $n$ is the number of nodes in the system.
% \qed
%\end{proof}